\documentclass{ifacconf}

\usepackage{graphicx}
\usepackage{natbib}

\usepackage{cite}
\usepackage{amsmath,amssymb,amsfonts}
\usepackage{textcomp}

\usepackage{mathrsfs}
\usepackage{mathtools}
\usepackage{stmaryrd}
\usepackage{bm}
\usepackage{cite}
\usepackage{booktabs}
\usepackage{algorithm}
\usepackage[noend]{algpseudocode}
\usepackage{physics}
\usepackage{subfigure}

\algrenewcommand\algorithmicindent{1em}

\makeatletter
\renewcommand{\ALG@name}{Protocol}
\makeatother
\algrenewcommand\algorithmicdo{}

\newcommand{\asmref}[1]{Assumption~\ref{#1}}

\newcommand{\lemref}[1]{Lemma~\ref{#1}}
\newcommand{\thmref}[1]{Theorem~\ref{#1}}

\newcommand{\secref}[1]{Section~\ref{#1}}
\newcommand{\appref}[1]{Appendix~\ref{#1}}
\newcommand{\figref}[1]{Fig.~\ref{#1}}

\newcommand{\proref}[1]{Protocol~\ref{#1}}

\newcommand{\N}{\mathbb{N}}
\newcommand{\Z}{\mathbb{Z}}
\newcommand{\Q}{\mathbb{Q}}
\newcommand{\R}{\mathbb{R}}
\newcommand{\ZZ}[1]{\Z_{\langle #1 \rangle}}
\newcommand{\QQ}[2]{\Q_{\langle #1, #2 \rangle}}
\newcommand{\Ecd}{\mathsf{Ecd}}
\newcommand{\Dcd}{\mathsf{Dcd}}
\newcommand{\Share}{\mathsf{Share}}
\newcommand{\Reconst}{\mathsf{Reconst}}
\newcommand{\Mult}{\mathsf{Mult}}
\newcommand{\Trunc}{\mathsf{Trunc}}

\DeclarePairedDelimiter{\floor}{\lfloor}{\rfloor}
\DeclarePairedDelimiter{\round}{\lceil}{\rfloor}
\DeclarePairedDelimiter{\share}{\llbracket}{\rrbracket}

\begin{document}

\thispagestyle{empty}
\hspace{-4.5mm}
\fbox{
\begin{minipage}{\textwidth-5mm}\scriptsize
© 2026 the authors.
This work has been accepted to IFAC for publication under a Creative Commons Licence CC-BY-NC-ND.
\end{minipage}
}
\newpage
\setcounter{page}{0}

\begin{frontmatter}

\title{Experimental Examination of Secure Two-Party Controller Computation\thanksref{footnoteinfo}} 

\thanks[footnoteinfo]{This work was supported by AFOSR Award FA9550-25-1-0347 and DARPA Award HR0011-25-3-0210.}

\author[First]{Kaoru Teranishi} 
\author[Second]{Jihoon Suh} 
\author[Second]{Takashi Tanaka}

\address[First]{Department of Information and Physical Sciences, Graduate School of Information Science and Technology, The University of Osaka, Osaka, Japan (e-mail: k-teranishi@ist.osaka-u.ac.jp)}
\address[Second]{School of Aeronautics and Astronautics, Purdue University, West Lafayette, IN 47907 USA (e-mail: suh95@purdue.edu, tanaka16@purdue.edu)}

\begin{abstract}
A secure two-party computation protocol for running dynamic controllers over secret sharing has recently been proposed.
Unlike encrypted control schemes based on homomorphic encryption, this protocol enables operating dynamic controllers for an infinite time horizon without controller-state decryption, controller-state reset, or input re-encryption.
However, the two-party setting introduces additional online communication between the computing parties, which may hinder real-time feasibility.
In this study, we demonstrate the feasibility of the protocol through implementation on a commercial cloud platform with an inverted pendulum testbed.
Experimental results show that the proposed protocol successfully stabilized the pendulum despite the online communication overhead.
\end{abstract}


\end{frontmatter}

\section{Introduction}

Encrypted control has emerged as a promising framework for enhancing the security of networked control systems \citep{Kogiso2015-go,Darup2021-qq}.
Beyond secure communication, it enables secure computation without revealing controller inputs, outputs, or parameters even when the controller computation is outsourced to an untrusted third party, such as a public cloud.
This is achieved by leveraging homomorphic encryption, which allows arithmetic operations to be performed directly on encrypted data.

A major challenge in encrypted control is avoiding overflow in the encrypted controller states when implementing dynamic controllers \citep{Cheon2018-vr}.
In a single-server setting, previous studies have attempted to address this problem by decrypting controller states at every time step \citep{Kogiso2015-go}, employing fully homomorphic encryption with bootstrapping \citep{Kim2016-lc}, periodically resetting controller states \citep{Murguia2020-ov}, or reformulating controller representations.
The controller-reformulation approach includes approximating a controller using a finite impulse response filter \citep{Schluter2021-ek}, transforming a controller via re-encrypted control inputs so that it has an integer state matrix \citep{Kim2023-pk}, and representing controller states using historical input and output data \citep{Teranishi2024-ha,Lee2025-jo}.
Although these approaches prevent overflow, the underlying difficulty stems from the intrinsic limitation that homomorphic encryption supports only addition and multiplication.

To overcome this limitation, a recent study has explored a two-party setting and applied secure multi-party computation to realize encrypted control \citep{Teranishi2025-wv}.
Instead of relying on homomorphic encryption, this approach employs secret sharing and constructs multiplication and truncation protocols.
The truncation protocol effectively prevents state overflow without requiring state decryption, state reset, or input re-encryption.
Compared with single-server methods, the two-party protocol simplifies system configuration, reduces the computational cost for the client, and improves security guarantees.
Despite these theoretical advantages, the two-party setting requires additional online communication between the servers.
This additional communication requirement raises questions about its practicality as the protocol has not been tested empirically.

Motivated by these concerns, this study experimentally examines the feasibility of the two-party controller computation protocol.
We begin by presenting a vectorized variant of the two-party computation to simplify the computation and yield a protocol more amenable to implementation.
Building on this variant, we implement the protocol on a commercial cloud platform and conduct experiments using an inverted pendulum testbed.
The experimental results demonstrate that the protocol successfully stabilizes the pendulum despite the online communication overhead imposed by the two-party setting.
Our main contributions are summarized as follows:
\begin{itemize}
    \item
    We clarify a sufficient condition for the modulus bit length in secret sharing that guarantees the vectorized two-party controller computation protocol avoids state overflow.
    This condition provides a practical guideline for selecting system parameters when implementing the protocol.

    \item
    We demonstrate the practicality of secure multi-party computation in an actual control system.
    To the best of our knowledge, this is the first implementation of secure multi-party computation for real-time control systems on a commercial computing service.

    \item
    We provide a detailed evaluation of the processing times of the subprotocols that compose the controller computation protocol.
    The measurements reveal the primary bottleneck in processing time and offer insights for further optimization of secure controller computation in real-time applications.
\end{itemize}

The remainder of this paper is organized as follows.
\secref{sec:protocol} introduces secret sharing, matrix triples, and the subprotocols for multiplication and truncation.
It also presents a vectorized variant of the two-party controller computation protocol.
\secref{sec:experiments} describes the experimental setup and results.
\secref{sec:conclusion} concludes the paper and discusses future work.

\emph{Notation:}
The sets of integers, positive integers, and real numbers are denoted by $\Z$, $\N$, and $\R$, respectively.
For $k \in \N$, $\ZZ{k} \coloneqq \{ -2^{k - 1}, \dots, 2^{k - 1} - 1 \}$ denotes the set of $k$-bit integers.
For $k, \ell \in \N$ such that $k > \ell$, $\QQ{k}{\ell} \coloneqq \{ 2^{-\ell} z \mid z \in \ZZ{k} \}$ denotes the set of $k$-bit fixed-point numbers with an $\ell$-bit fractional part.
For $q \ge 2$, $\Z_q \coloneqq \Z \cap [-q/2, q/2)$ denotes the set of integers within $[-q/2, q/2)$.
For $z \in \Z$, $z \bmod q \coloneqq z - \floor{ \frac{z + q / 2}{q} } q$ denotes the reduction of $z$ modulo $q$ over $\Z_q$, where $\floor{x} \coloneqq \max\{z \in \Z \mid z \le x\}$ is the floor function.
The rounding function is also defined by $\round{x} \coloneqq \floor{x + 1/2}$.
For vectors and matrices, $(\cdot \bmod q)$, $\floor{ \cdot }$, and $\round{ \cdot }$ are performed element-wise.

\section{Two-Party Controller Computation}\label{sec:protocol}

This section presents a vectorized variant of the two-party protocol for controller computation.
Consider the discrete-time linear time-invariant controller
\[
    \begin{aligned}
        x_{t+1} &= A x_t + B y_t, \\
        u_t &= C x_t + D y_t,
    \end{aligned}
\]
where $t = 0, 1, 2, \dots$ is the time index, $x_t \in \R^n$ is the state, $u_t \in \R^m$ is the control input, and $y_t \in \R^p$ is the sensor measurement.
The controller can be represented in the form of a matrix-vector product,
\begin{equation}
    \psi_t = \Phi \xi_t,
    \label{eq:controller}
\end{equation}
where
\[
    \psi_t \coloneqq
    \begin{bmatrix}
        x_{t+1} \\
        u_t
    \end{bmatrix}, \quad
    \Phi \coloneqq
    \begin{bmatrix}
        A & B \\
        C & D
    \end{bmatrix}, \quad
    \xi_t \coloneqq
    \begin{bmatrix}
        x_t \\
        y_t
    \end{bmatrix}.
\]
We assume that the controller parameters, initial state, and sensor measurement at each time step are given by fixed-point numbers.
Note that this assumption is reasonable in practice because real numbers can be approximated by fixed-point numbers with desired precision.

\begin{assum}\label{asm:fixed-point}
    The parameter matrix, initial state, and sensor measurement in \eqref{eq:controller} satisfy $\Phi \in \QQ{k}{\ell}^{(n + m) \times (n + p)}$, $x_0 \in \QQ{k}{\ell}^n$, and $y_t \in \QQ{k}{\ell}^p$.
\end{assum}

Our goal is to securely implement the controller computation \eqref{eq:controller} on untrusted servers.
To achieve this, we employ a two-party protocol in \citet{Teranishi2025-wv} with modifications.
In what follows, we denote the computing parties as $P_i$, $i \in \{0, 1\}$.

\subsection{Secret Sharing}

The two-party protocol guarantees security by splitting the controller \eqref{eq:controller} into two randomized controllers and computing each of the controllers separately \citep{Teranishi2025-wv}.
Such randomization is performed using a cryptographic primitive called secret sharing.

\begin{defn}
    A (vectorized) secret sharing scheme over $\Z_q$ with a prime $q$ consists of two polynomial-time algorithms $\Share$ and $\Reconst$:
    \begin{itemize}
        \item $\share{X} \gets \Share(X)$:
        The share generation algorithm takes a secret $X \in \Z_q^{d_1 \times d_2}$ as input and outputs shares $\share{X} = (\share{X}_0, \share{X}_1) = (R, X - R \bmod q)$, where $R \in \Z_q^{d_1 \times d_2}$ is a random matrix.
        
        \item $X \gets \Reconst(\share{X})$:
        The reconstruction algorithm takes shares $\share{X}$ as input and outputs $X = \share{X}_0 + \share{X}_1 \bmod q$.
    \end{itemize}
\end{defn}

In our two-party setting, the party $P_i$ receives the share $\share{X}_i$.
Each share $\share{X}_i$ is distributed over $\Z_q^{d_1 \times d_2}$ uniformly at random, and each share reveals no information about a secret $X$ to $P_i$.
The secret $X$ can be reconstructed by combining both shares, i.e.,
\[
    \Reconst(\Share(X)) = X.
\]
Moreover, by construction, the secret sharing scheme allows (constant) addition/subtraction,
\begin{align*}
    \share{X} \pm Y &\coloneqq (\share{X}_0 \pm Y \bmod q, \share{X}_1), \\
    \share{X} \pm \share{Y} &\coloneqq (\share{X}_0 \pm \share{Y}_0 \bmod q, \share{X}_1 \pm \share{Y}_1 \bmod q),
\end{align*}
and constant multiplication,
\begin{align*}
    k \share{X} &\coloneqq (k \share{X}_0 \bmod q, k \share{X}_1 \bmod q), \\
    \share{X} Z &\coloneqq (\share{X}_0 Z \bmod q, \share{X}_1 Z \bmod q),
\end{align*}
where $X, Y \in \Z_q^{d_1 \times d_2}$, $Z \in \Z_q^{d_2 \times d_3}$, and $k \in \Z_q$.
The computations can be performed locally and satisfy
\begin{align*}
    \share{X} \pm Y &= \share{X} \pm \share{Y} = \share{X \pm Y \bmod q}, \\
    k \share{X} &= \share{kX \bmod q}, \\
    \share{X} Z &= \share{XZ \bmod q}.
\end{align*}
We also define $X \pm \share{Y}$ and $X \share{Z}$ in a similar way.

\subsection{Matrix Triple}

Recall that we aim to implement the matrix-vector product \eqref{eq:controller} over secret sharing.
Unlike addition and subtraction, however, matrix multiplication between two shares $\share{X}$ and $\share{Y}$ (of appropriate dimensions) is not straightforward.
This is because it requires computing cross terms $\share{X}_0 \share{Y}_1$ and $\share{X}_1 \share{Y}_0$, which include all shares of both $X$ and $Y$, and thus discloses the secrets to the parties.
To overcome this difficulty, the two-party protocol applies Beaver triples \citep{Beaver2007-kr} to construct a subprotocol for multiplication.
We use matrix triples in \citet{Mohassel2017-yi}, which are vectorized Beaver triples.
Here, a matrix triple over $\Z_q$ is a triplet $(\share{U}, \share{V}, \allowbreak \share{W})$ that comprises the shares of random matrices $U \in \Z_q^{d_1 \times d_2}$, $V \in \Z_q^{d_2 \times d_3}$, and $W \in \Z_q^{d_1 \times d_3}$ such that $W = UV \bmod q$.

Let $X \in \Z_q^{d_1 \times d_2}$ and $Y \in \Z_q^{d_2 \times d_3}$.
\proref{pro:mult} computes the shares $\share{Z}$ of $Z = X Y \bmod q \in \Z_q^{d_1 \times d_3}$ using a matrix triple.
Suppose the party $P_i$ has $\share{X}_i$, $\share{Y}_i$, and $(\share{U}_i,  \share{V}_i, \share{W}_i)$.
Each party randomizes its own shares $\share{X}_i$ and $\share{Y}_i$ as $\share{S}_i = \share{X}_i - \share{U}_i \bmod q$ and $\share{T}_i = \share{Y}_i - \share{V}_i \bmod q$, respectively.
The parties then open the randomized shares and reconstruct $S$ and $T$.
Each party computes and outputs $\share{Z}_i = \share{U}_i T + S \share{V}_i + \share{W}_i + i \cdot ST \bmod q$.
The output shares satisfy
\begin{align*}
    &\Reconst(\share{Z})
    = \share{Z}_0 + \share{Z}_1 \bmod q, \\
    &= U T + S V + W + ST \bmod q, \\
    &= U (Y \!-\! V) \!+\! (X \!-\! U) V \!+\! U V \!+\! (X \!-\! U)(Y \!-\! V) \bmod q, \\
    &= XY \bmod q,
\end{align*}
which implies $\share{Z} = \share{XY \bmod q}$.
We denote $\share{X} \share{Y} \coloneqq \Mult(\share{X}, \share{Y})$ for notational convenience.

\begin{figure}[t]
    \begin{algorithm}[H]
        \caption{Matrix multiplication ($\Mult$)}
        \label{pro:mult}
        \begin{algorithmic}[1]
            \Require Shares $\share{X}$ of $X \in \Z_q^{d_1 \times d_2}$, shares $\share{Y}$ of $Y \in \Z_q^{d_2 \times d_3}$, and a matrix triple $(\share{U}, \share{V}, \share{W})$.
            \Ensure Shares $\share{Z} = \share{XY \bmod q}$.
            \State $\share{S}_i \gets \share{X}_i - \share{U}_i \bmod q$, $\share{T}_i \gets \share{Y}_i - \share{V}_i \bmod q$
            \State $S \gets \Reconst(\share{S})$, $T \gets \Reconst(\share{T})$
            \State $\share{Z}_i \gets \share{U}_i T + S \share{V}_i + \share{W}_i + i \cdot ST \bmod q$
        \end{algorithmic}
    \end{algorithm}
    \vspace{-.5em}
\end{figure}

\subsection{Bit Truncation}

\asmref{asm:fixed-point} implies that the parameter matrix and input/output vectors in \eqref{eq:controller} are represented by fixed-point numbers.
To apply the multiplication protocol to $X \in \QQ{k}{\ell}^{d_1 \times d_2}$ and $Y \in \QQ{k}{\ell}^{d_2 \times d_3}$, we need to encode them into integer matrices over $\Z_q$.
When $\ZZ{k} \subset \Z_q$, the encoding and corresponding decoding processes are respectively performed by
\[
    \Ecd(X, \ell) \coloneqq \bar{X} = 2^\ell X, \quad
    \Dcd(\bar{X}, \ell) \coloneqq 2^{-\ell} \bar{X}.
\]
The multiplication of shares of $\bar{X} = \Ecd(X, \ell)$ and $\bar{Y} = \Ecd(Y, \ell)$ is then given by $\share{\bar{X}} \share{\bar{Y}} = \share{2^{2 \ell} X Y \bmod q}$.
This increases the fractional part of the resulting matrix by $\ell$~bits.
The controller state $x_t$ is recursively multiplied by $A$ in the computation \eqref{eq:controller}, and thus, the fractional part of the encoded state reaches the data size $k$ (i.e., overflow) at a certain finite time step whenever $A$ is not an integer matrix \citep{Cheon2018-vr}.

To avoid this issue, the two-party protocol adopted the truncation protocol in \citet{Escudero2020-ih}.
A vectorized version of the truncation protocol is shown in \proref{pro:trunc}.
Let $\lambda \in \N$ be a security parameter of the protocol.
Suppose that $\kappa = \floor{\log_2 q} - \lambda - 1 > \ell$, and the party $P_i$ has the share $\share{X}_i$ of $X \in \ZZ{\kappa}^{d_1 \times d_2}$ and the pair of shares $(\share{R}_i, \share{R'}_i)$ of random matrices $R \in \ZZ{\kappa - \ell + \lambda}^{d_1 \times d_2}$ and $R' \in \ZZ{\ell}^{d_1 \times d_2}$.
Each party randomizes its own share $\share{X}_i$ as $\share{Y}_i = \share{X}_i + 2^\ell \share{R}_i + \share{R'}_i + i \cdot 2^{\ell - 1} J \bmod q$, where $J$ is the $d_1$-by-$d_2$ matrix of ones.
The party $P_0$ sends its randomized share to $P_1$, and then $P_1$ reconstructs $Y$.
Each party computes and outputs $\share{Z}_i = 2^{-\ell} ( \share{X}_i + \share{R'}_i - i \cdot ( Y - 2^{\ell - 1} J \bmod 2^\ell ) ) \bmod q$.
Here, note that $2^{-\ell}$ is the modular inverse of $2^\ell$ modulo $q$.
As a result, the output shares represent shares of a matrix obtained by truncating the least $\ell$ bits of $X$, with an error $W \in \{-1, 0, 1\}^{d_1 \times d_2}$.

\begin{lem}\label{lem:trunc}
    Suppose $q$ is a prime.
    Let $\kappa, \ell, \lambda \in \N$ such that $\kappa = \floor{ \log_2 q } - \lambda - 1 > \ell$.
    It holds that, for all $X \in \ZZ{\kappa}^{d_1 \times d_2}$,
    \[
        \Reconst\qty(\Trunc\qty(\share{X}, \ell)) = \round*{ \frac{ X }{ 2^\ell } } + W,
    \]
    where $\share{X} \gets \Share(X)$ and $W \in \{-1, 0, 1\}^{d_1 \times d_2}$.
\end{lem}

\begin{pf}
    The proof is obtained by adopting Lemma~1 in \citet{Teranishi2025-wv} for each element of $X$.
\end{pf}

\begin{figure}[t]
    \begin{algorithm}[H]
        \caption{Bit truncation ($\Trunc$)}
        \label{pro:trunc}
        \begin{algorithmic}[1]
            \Require Shares $\share{X}$ of $X \in \ZZ{\kappa}^{d_1 \times d_2}$, shares $(\share{R}, \share{R'})$ of random matrices $R \in \ZZ{\kappa - \ell + \lambda}^{d_1 \times d_2}$ and $R' \in \ZZ{\ell}^{d_1 \times d_2}$, and bit length $\ell \in \N$, where $\kappa = \floor{ \log_2 q } - \lambda - 1 > \ell$, and $\lambda \in \N$ is a security parameter.
            \Ensure Shares $\share{Z} = \share{\round{X / 2^\ell} + W}$ with $W \in \Z_3^{d_1 \times d_2}$.
            \State $\share{Y}_i \gets \share{X}_i + 2^\ell \share{R}_i + \share{R'}_i + i \cdot 2^{\ell - 1} J \bmod q$
            \State $Y \gets \Reconst(\share{Y})$
            \State $\share{Z}_i \gets 2^{-\ell} ( \share{X}_i + \share{R'}_i - i \cdot ( Y - 2^{\ell - 1} J \bmod 2^\ell ) ) \bmod q$
        \end{algorithmic}
    \end{algorithm}
    \vspace{-1em}
\end{figure}

\subsection{Controller Computation}

Using the multiplication and truncation protocols, we now present a controller computation protocol (\proref{pro:controller}).
The protocol consists of the offline phase (lines 1--4) and online phase (lines 5--16).
\figref{fig:controller} illustrates the data flow of the protocol, where the dotted arrows represent communication of the intermediate messages $(\share{S_t}_i, \share{T_t}_i)$ and $\share{Y_t}_0$ in the subprotocols.
We refer to the entities corresponding to the gray, blue, and orange boxes as a client, $P_0$, and $P_1$, respectively.
In what follows, the dimensions $d_1$, $d_2$, and $d_3$ in \proref{pro:mult} are set to $n + m$, $n + p$, and $1$, respectively.
The dimensions $d_1$ and $d_2$ in \proref{pro:trunc} are also set to $n$ and $1$, respectively.

\begin{figure}[t]
    \begin{algorithm}[H]
        \caption{Two-party controller computation}
        \label{pro:controller}
        \begin{algorithmic}[1]
            \Require Controller parameter $\Phi \in \QQ{k}{\ell}^{(n + m) \times (n + p)}$, initial state $x_0 \in \QQ{k}{\ell}^n$, sensor measurement $y_t \in \QQ{k}{\ell}^p$, bit lengths $k, \ell \in \N$, modulus $q$, and security parameter $\lambda \in \N$.
            \Ensure Control input $u_t \in \R^m$.
            \State $\triangleright$ Client
            \State $\bar{\Phi} \gets \Ecd(\Phi, \ell)$, $\share{\bar{\Phi}} \gets \Share(\bar{\Phi})$
            \State $\bar{x}_0 \gets \Ecd(x_0, \ell)$, $\share{\bar{x}_0} \gets \Share(\bar{x}_0)$
            \State Send $\share{\bar{\Phi}}_i$ and $\share{\bar{x}_0}_i$ to $P_i$
            \ForAll{$t = 0, 1, 2, \dots$}
                \State $\triangleright$ Client
                \State $\bar{y}_t \gets \Ecd(y_t, \ell)$, $\share{\bar{y}_t} \gets \Share(\bar{y}_t)$
                \State Send $\share{\bar{y}_t}_i$ to $P_i$
                \State Generate $(\share{U_t}, \share{V_t}, \share{W_t})$ and $(\share{R_t}, \share{R'_t})$
                \State Send $(\share{U_t}_i, \share{V_t}_i, \share{W_t}_i)$, and $(\share{R_t}_i, \share{R'_t}_i)$ to $P_i$
                \State $\triangleright$ Parties
                \State $\share{\bar{\xi}_t} \gets [ \share{\bar{x}_t}^\top \ \share{\bar{y}_t}^\top ]^\top$, $\share{\bar{\psi}_t} \gets \Mult(\share{\bar{\Phi}}, \share{\bar{\xi}_t})$
                \State $[ \share{\tilde{x}_{t+1}}^\top \ \share{\bar{u}_t}^\top ]^\top \gets \share{\bar{\psi}_t}$, $\share{\bar{x}_{t+1}} \gets \Trunc(\share{\tilde{x}_{t+1}}, \ell)$
                \State Send $\share{\bar{u}_t}$ to the client
                \State $\triangleright$ Client
                \State $\bar{u}_t \gets \Reconst(\share{\bar{u}_t})$, $u_t \gets \Dcd(\bar{u}_t, 2\ell)$
            \EndFor
        \end{algorithmic}
    \end{algorithm}
    \vspace{-.5em}
\end{figure}

\begin{figure}[t]
    \centering
    \includegraphics[scale=.9]{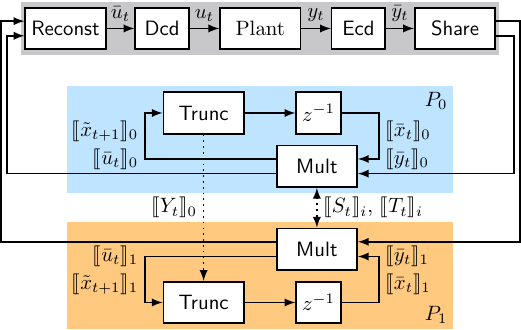}
    \vspace{-.5em}
    \caption{Data flow of \proref{pro:controller}.}
    \label{fig:controller}
\end{figure}

In the offline phase, the client generates shares of the parameter matrix $\Phi$ and the initial state $x_0$ as
\begin{equation}
    \share{\bar{\Phi}} \gets \Share(\bar{\Phi}), \quad \share{\bar{x}_0} \gets \Share(\bar{x}_0),
    \label{eq:controller-params}
\end{equation}
where
\begin{align*}
    \bar{\Phi} &\coloneqq \Ecd(\Phi, \ell) = 2^\ell \Phi \in \ZZ{k}^{(n + m)(n + p)}, \\
    \bar{x}_0 &\coloneqq \Ecd(x_0, \ell) = 2^\ell x_0 \in \ZZ{k}^n.
\end{align*}
The client then sends $\share{\bar{\Phi}}_i$ and $\share{\bar{x}_0}_i$ to $P_i$ who stores the received shares.

In the online phase, for every $t$, the client reads the sensor measurement $y_t$ and generates its shares by
\begin{equation}
    \share{\bar{y}_t} \gets \Share(\bar{y}_t), \quad \bar{y}_t \coloneqq \Ecd(y_t, \ell) = 2^\ell y_t \in \ZZ{k}^p.
    \label{eq:controller-sensor}
\end{equation}
Here, note that $\bar{y}_t$ is always a $k$-bit integer vector thanks to \asmref{asm:fixed-point}.
It also generates a matrix triple $(\share{U_t}, \share{V_t}, \share{W_t})$ and a pair of random shares $(\share{R_t}, \share{R'_t})$.
The client sends $\share{\bar{y}_t}_i$, $(\share{U_t}_i, \share{V_t}_i, \share{W_t}_i)$, and $(\share{R_t}_i, \share{R'_t}_i)$ to $P_i$.
The parties compute
\begin{equation}
    \share{\bar{\psi}_t} \gets \Mult(\share{\bar{\Phi}}, \share{\bar{\xi}_t}), \quad \share{\bar{\xi}_t} \coloneqq
    \begin{bmatrix}
        \share{\bar{x}_t} \\
        \share{\bar{y}_t}
    \end{bmatrix}
    \label{eq:controller-mult}
\end{equation}
by \proref{pro:mult} using the matrix triple.
Then, the parties compute
\begin{equation}
    \share{\bar{x}_{t+1}} \gets \Trunc(\share{\tilde{x}_{t+1}}, \ell), \quad
    \begin{bmatrix}
        \share{\tilde{x}_{t+1}} \\
        \share{\bar{u}_t}
    \end{bmatrix}
    \coloneqq \share{\bar{\psi}_t}
    \label{eq:controller-trunc}
\end{equation}
by \proref{pro:trunc} using the random shares.
Finally, each party $P_i$ separately returns $\share{\bar{u}_t}_i$ to the client who obtains the control input by
\begin{equation}
    u_t = \Dcd(\bar{u}_t, 2 \ell), \quad \bar{u}_t \coloneqq \Reconst(\share{\bar{u}_t}).
    \label{eq:controller-input}
\end{equation}
The following theorem provides a bit length of the modulus $q$ required for the correct operation of \proref{pro:controller}.

\begin{thm}\label{thm:modulus-bits}
    Consider the controller \eqref{eq:controller} under \asmref{asm:fixed-point}.
    Assume that $A$ is Schur stable.
    Let $\lambda \in \N$ be a security parameter.
    If the bit lengths $k, \ell$ and the modulus $q$ satisfy $k - \ell \ge 2$ and
    \begin{align*}
        \log_2 q &\ge 3k - \ell + \lambda + 2 + \floor{ \log_2 f(n, p, c, \gamma) }, \\
        f(n, p, c, \gamma) &= \max\qty{n, p} (1 + p) \sqrt{n} \frac{ c }{ 1 - \gamma },
    \end{align*}
    then the controller that consists of \eqref{eq:controller-params}, \eqref{eq:controller-mult}, and \eqref{eq:controller-trunc} holds
    \begin{align*}
        \bar{x}_{t+1} &= \round{ A \bar{x}_t + B \bar{y}_t } + w_t, \\
        \bar{u}_t &= \bar{C} \bar{x}_t + \bar{D} \bar{y}_t,
    \end{align*}
    for all $t = 0, 1, 2, \dots$, where $\bar{x}_t = \Reconst(\share{\bar{x}_t})$, $\bar{y}_t$ and $\bar{u}_t$ are respectively defined in \eqref{eq:controller-sensor} and \eqref{eq:controller-input}, $w_t \in \{-1, 0, 1\}^n$, $c \ge 1$ and $\gamma \in (0, 1)$ are constants such that $\norm{ A^t } \le c \gamma^t$ for all $t$, and $\norm{\cdot}$ is the induced $2$-norm.
\end{thm}

\begin{pf}
    See \appref{app:proof}.
\end{pf}

The theorem guarantees that, when choosing a sufficiently large modulus $q$, \proref{pro:controller} securely runs the encoded version of \eqref{eq:controller} for an infinite time horizon without state decryption, state reset, or input re-encryption.
The selection of $q$ involves a trade-off, as a larger $q$ allows for increasing the fixed-point precision or the security level, although it increases the communication overhead.
While the encoded controller is not equivalent to the original one due to the rounding error and $w_t$ imposed by the truncation protocol, their effects can be negligible in terms of control performance.
We will confirm this by experiment in the next section.

\begin{rem}[Security]
    \proref{pro:controller} inherits the security guarantees of the original two-party controller computation protocol.
    The communication channels are protected via standard symmetric-key encryption, preventing any data leakage to network eavesdroppers.
    Moreover, the computing parties obtain no information beyond their input and output shares as long as they are semi-honest and do not collude, thereby achieving $\lambda$-bit statistical security.
    Statistical security is a stronger security notion than computational security because it does not restrict adversary's computational capability to polynomial time.
    See \citet{Teranishi2025-wv} for more details on the protocol security.
\end{rem}

\section{Experiments}\label{sec:experiments}

This section examines the controller computation protocol through experiments conducted on a laboratory testbed.
\figref{fig:system} illustrates the system architecture of the developed testbed.
The controlled plant was a rotary inverted pendulum (Quanser Inc.), shown in \figref{fig:rip}, connected to a laptop computer (Windows~11, Intel(R) Core(TM) Ultra~7 165U, 32~GB RAM) as the client.
Two Amazon EC2 instances (t3.micro, us-east-2) were deployed as the parties $P_0$ and $P_1$, and the protocol was implemented using Python~3.12 on these servers.
The laptop computer communicated with the servers via TCP over a Wi-Fi connection through an access point.

\begin{figure}[t]
    \centering
    \includegraphics[scale=.8]{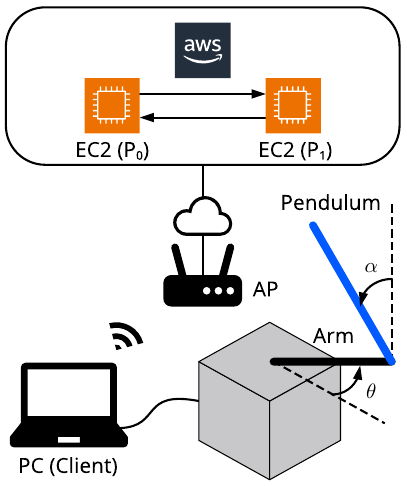}
    \vspace{-.5em}
    \caption{System architecture of the testbed.}
    \label{fig:system}
\end{figure}

\begin{figure}[t]
    \centering
    \includegraphics[height=40mm]{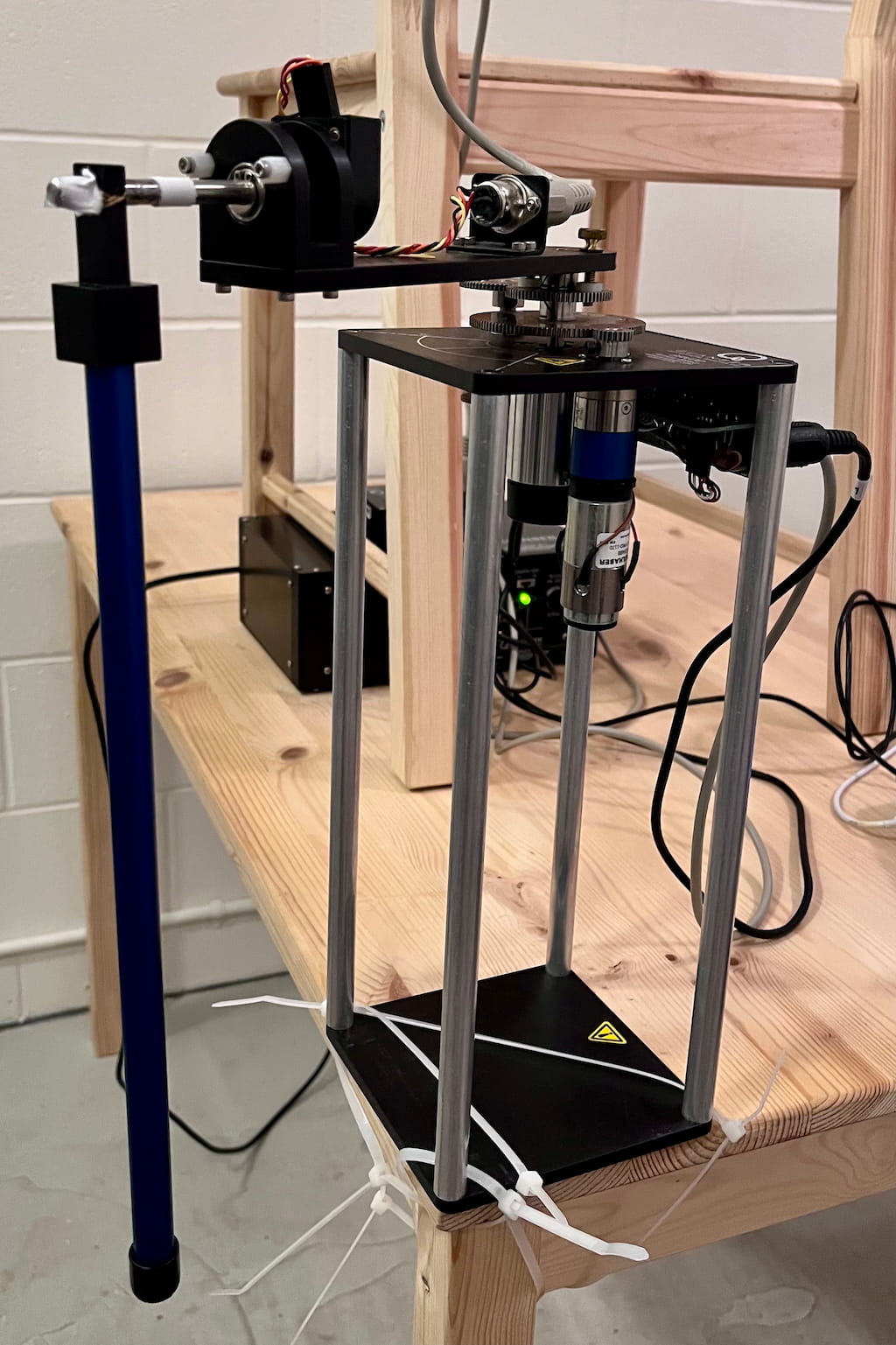}
    \vspace{-.5em}
    \caption{Rotary inverted pendulum.}
    \label{fig:rip}
\end{figure}

The objective of the experiment was to stabilize the pendulum around its upright equilibrium.
Let $\theta$ and $\alpha$ denote the encoder measurements of the arm and pendulum angles, respectively.
Given that the encoder resolutions are $4096$ counts per revolution, we assumed $y_t = [\theta_t \ \alpha_t]^\top \in (-180, 180]^2$ and $y_t \in \QQ{18}{9}^2$.
The sampling time was set to $40$~ms, and the plant was treated as a one-step input-delayed system, as we verified that the communication delay remained within this interval in our implementation.
That is, the control input computed at time step $t$ was applied to the plant at the next time step $t + 1$.
Here, the control input corresponds to the voltage applied to the motor mounted on the arm.

We designed the controller as
\begin{alignat*}{3}
    A &=
    \begin{bmatrix}
        0 & 0 & 0 \\
        0 & 0 & 0 \\
        k_3 & k_4 & k_5
    \end{bmatrix}, &\quad
    B &=
    \begin{bmatrix}
        -50 & 0 \\
        0 & -50 \\
        k_1 & k_2
    \end{bmatrix}, &\quad \\
    C &=
    \begin{bmatrix}
        k_3 & k_4 & k_5
    \end{bmatrix}, &\quad
    D &=
    \begin{bmatrix}
        k_1 & k_2
    \end{bmatrix}, &\quad
    x_0 &=
    \begin{bmatrix}
        0 & 0 & 0
    \end{bmatrix}^\top,
\end{alignat*}
where
\begin{alignat*}{2}
    k_1 &=  33.022125244140625,  &\ 
    k_2 &= -51.49261474609375,   \\
    k_3 &=   0.6586151123046875, &\ 
    k_4 &=  -0.7884063720703125, \\
    k_5 &=  -0.5822296142578125. &&
\end{alignat*}
The controller consists of pseudo-differentiators for estimating the angular velocities of the arm and the pendulum, denoted by $\dot{\theta}$ and $\dot{\alpha}$, respectively, and the state-feedback
\[
    u_t = k_1 \theta_t + k_2 \alpha_t + k_3 \dot{\theta}_t + k_4 \dot{\alpha}_t + k_5 u_{t-1}.
\]
Since all controller parameters are in $\QQ{23}{16}$, the bit lengths were set to $k = 64$ and $\ell = 32$ to satisfy \asmref{asm:fixed-point}.
Furthermore, the bit length of $q$ was set to $256$~bits, satisfying the lower bound of $248$~bits required by \thmref{thm:modulus-bits} with $c = 2.34$, $\gamma = 0.59$, and $\lambda = 80$.

Figs.~\ref{fig:u} and \subref{fig:alpha} show the control input and the pendulum angle in the experiment, respectively.
After activating \proref{pro:controller}, we manually moved the pendulum from the downright position to the upright position.
The automatic control began at approximately $3.5$~s, as indicated by the dashed line.
These results confirm that the proposed protocol successfully stabilizes the pendulum around its upright equilibrium using the commercial cloud platform even with rounding errors and the errors imposed by the truncation protocol.

\begin{figure}[t]
    \centering
    \subfigure[Input voltage.]{\includegraphics[scale=.8]{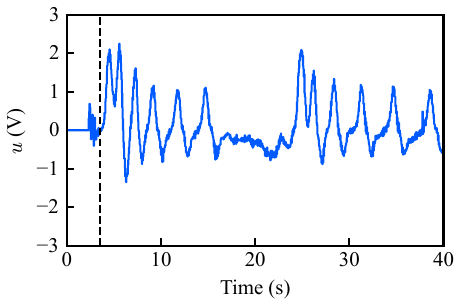}\label{fig:u}}
    \subfigure[Pendulum angle.]{\includegraphics[scale=.8]{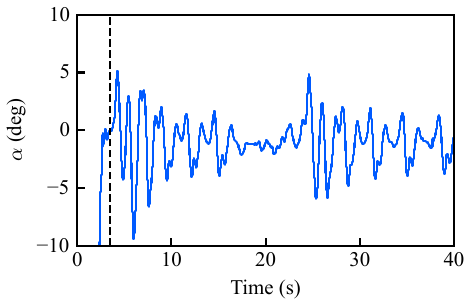}\label{fig:alpha}}
    \vspace{-.5em}
    \caption{Experimental results.}
    \label{fig:control}
\end{figure}

Next, we evaluate the processing times of the subprotocols.
For each protocol, we measured the round-trip time from when the client sent the input shares to when it received the output shares.
Each measurement was repeated 100~times, and the minimum, mean, and maximum times were recorded.
\figref{fig:mult-time} depicts the processing times of \proref{pro:mult} for dimensions $d_1 = d_2 = 10, 20, \dots, 100$.
Similarly, \figref{fig:trunc-time} depicts the processing times of \proref{pro:trunc} for dimensions $d_1 = 10, 20, \dots, 100$.
In both figures, the light green shaded areas represent the $99$~\% confidence intervals.
The results indicate that the processing times of both subprotocols remain on the order of milliseconds even for large dimensions.
Furthermore, although the processing time of \proref{pro:mult} increases with the dimensions, the processing time of \proref{pro:trunc} remains almost constant.
These findings suggest that the proposed protocol is applicable to systems with shorter sampling times, provided that the implementation of \proref{pro:mult} is optimized and wired communication is used to reduce the communication delays between the client and the parties.

\begin{figure}
    \centering
    \subfigure[\proref{pro:mult}.]{\includegraphics[scale=.8]{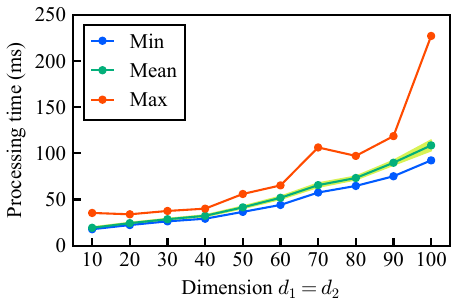}\label{fig:mult-time}}
    \subfigure[\proref{pro:trunc}.]{\includegraphics[scale=.8]{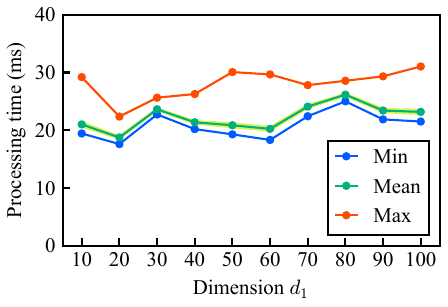}\label{fig:trunc-time}}
    \vspace{-.5em}
    \caption{Processing times of the subprotocols.}
    \label{fig:time}
\end{figure}

\section{Conclusion}\label{sec:conclusion}

We presented an experimental evaluation of a secure two-party computation protocol for dynamic controllers.
The protocol was implemented on a commercial cloud platform and validated using an inverted pendulum testbed.
The experimental results demonstrated that the protocol successfully stabilized the pendulum despite the online communication overhead between the computing parties.
Furthermore, the processing times of the multiplication and truncation subprotocols remained on the order of milliseconds even for large input sizes, with the primary bottleneck arising from the multiplication protocol.
These results confirm the real-time feasibility of the protocol under the experimental conditions.
Future work includes optimizing the protocol implementation in C/C++ and evaluating the achievable control performance in practical systems.

\appendix
\section{Proof of \thmref{thm:modulus-bits}}\label{app:proof}

Note first that the constants $c$ and $\gamma$ exist if $A$ is Schur stable, and they satisfy $c / (1 - \gamma) > 1$.
Let $\kappa = \floor{\log_2 q} - \lambda - 1$, then $\ZZ{\kappa} \subset \Z_q$ and $\kappa \ge 2 k + (k - \ell) + \log_2 f(n, p, c, \gamma) > \ell$, where $k - \ell \ge 1$.
If $\bar{A} \bar{x}_t + \bar{B}_t \bar{y}_t \in \ZZ{\kappa}^n$ and $\bar{C} \bar{x}_t + \bar{D}_t \bar{y}_t \in \ZZ{\kappa}^p$, \lemref{lem:trunc} implies that
\begin{align*}
    \bar{x}_{t+1}
    &= \Reconst(\Trunc(\share{\tilde{x}_{t+1}}, \ell))
    = \round{ A \bar{x}_t + B \bar{y}_t } + w_t, \\
    \bar{u}_t
    &= \Reconst(\share{\bar{u}_t})
    = \bar{C} \bar{x}_t + \bar{D} \bar{y}_t,
\end{align*}
where $\share{\tilde{x}_{t+1}} = \share{\bar{A}} \share{\bar{x}_t} + \share{\bar{B}} \share{\bar{y}_t} = \share{\bar{A} \bar{x}_t + \bar{B} \bar{y}_t}$ and $\share{\bar{u}_t} = \share{\bar{C}} \share{\bar{x}_t} + \share{\bar{D}} \share{\bar{y}_t} = \share{\bar{C} \bar{x}_t + \bar{D} \bar{y}_t}$ hold since $\ZZ{\kappa} \subset \Z_q$, and $w_t \in \{-1, 0, 1\}^n$.
Hence, the claim holds if $\bar{A} \bar{x}_t + \bar{B} \bar{y}_t \in \ZZ{\kappa}^n$ and $\bar{C} \bar{x}_t + \bar{D} \bar{y}_t \in \ZZ{\kappa}^p$ for all $t$.

For $t = 0$, it follows that
\begin{align*}
    & \max\qty{ \norm{\bar{A} \bar{x}_0 + \bar{B} \bar{y}_0}_\infty, \norm{\bar{C} \bar{x}_0 + \bar{D} \bar{y}_0}_\infty } \\
    &\le 2^{k - 1} \max\qty{ n, p } (\norm{\bar{x}_0}_\infty + \norm{\bar{y}_0}_\infty), \\
    &\le 2^{2k - 1} \max\{n, p\}, \\
    &< 2^{2k + (k - \ell) - 2} f(n, p, c, \gamma) = 2^{\kappa - 2} < 2^{\kappa - 1} - 1,
\end{align*}
which means $\bar{A} \bar{x}_0 + \bar{B} \bar{y}_0 \in \ZZ{\kappa}^n$ and $\bar{C} \bar{x}_0 + \bar{D} \bar{y}_0 \in \ZZ{\kappa}^p$.
Here, $\norm{\cdot}_\infty$ denotes the infinity norm.

Suppose that $\bar{A} \bar{x}_t + \bar{B} \bar{y}_t \in \ZZ{\kappa}^n$ and $\bar{C} \bar{x}_t + \bar{D} \bar{y}_t \in \ZZ{\kappa}^p$ hold for all $t = 0, \dots, \tau - 1$ with $\tau \in \N$.
This yields the following expression for $\bar{x}_\tau$,
\begin{align*}
    \bar{x}_\tau
    &= \round{ A \bar{x}_{\tau - 1} + B \bar{y}_{\tau - 1}} + w_{\tau - 1}, \\
    &= A^\tau \bar{x}_0 + \sum_{s=0}^{\tau - 1} A^s B \bar{y}_{\tau - 1 - s} + \sum_{s=0}^{\tau - 1} A^s (e_{\tau - 1 - s} + w_{\tau - 1 - s}),
\end{align*}
where $e_{\tau - 1 - s} \in (-1/2, 1/2]^n$.
It then holds that
\begin{align*}
    \norm{\bar{x}_\tau}_\infty
    &\le \qty( 2^{k - 1} + 2^{k - \ell - 1} p \cdot 2^{k - 1} + \frac{3}{2} ) \norm{ \sum_{s=0}^\tau A^s }_\infty, \\
    &< 2^{2(k - 1) - \ell} (1 + p) \sqrt{n} \frac{ c }{ 1 - \gamma },
\end{align*}
where we used $k - \ell \ge 2$ and $\norm{ \sum_{s=0}^\tau A^s }_\infty \le \sqrt{n} \norm{ \sum_{s=0}^\tau A^s } \allowbreak < \sqrt{n} \sum_{s=0}^\infty \norm{ A^s } = c \sqrt{n} / (1 - \gamma)$.
This implies that
\begin{align*}
    & \max\qty{ \norm{\bar{A} \bar{x}_\tau + \bar{B} \bar{y}_\tau}_\infty, \norm{\bar{C} \bar{x}_\tau + \bar{D} \bar{y}_\tau}_\infty } \\
    &\le 2^{k - 1} \max\qty{ n, p } (\norm{\bar{x}_\tau}_\infty + \norm{\bar{y}_\tau}_\infty), \\
    &< 2^{k - 1} \max\{n, p\} \qty( 2^{2(k - 1) - \ell} (1 + p) \sqrt{n} \frac{ c }{ 1 - \gamma } + 2^{k - 1} ), \\
    &< 2^{2k + (k - \ell) - 2} f(n, p, c, \gamma) < 2^{\kappa - 1} - 1,
\end{align*}
which concludes $\bar{A} \bar{x}_\tau + \bar{B} \bar{y}_\tau \in \ZZ{\kappa}^n$ and $\bar{C} \bar{x}_\tau + \bar{D} \bar{y}_\tau \in \ZZ{\kappa}^p$.

\bibliography{ifacconf}

\end{document}